\documentclass[12pt]{article}
\usepackage{epsfig}

\usepackage{amsthm}
\usepackage{amsfonts}
\usepackage{amssymb}
\usepackage{graphicx}
\usepackage{amsmath}
\usepackage{lscape}
\usepackage{lscape,color}

\definecolor{c30}{rgb}{0,.7,.2}

\definecolor{d30}{rgb}{1,0,0}

\newcommand{\bea}{\begin{eqnarray}}
\newcommand{\eea}{\end{eqnarray}}
\newcommand{\nn}{\nonumber}
\newcommand{\be}{\begin{equation}}
\newcommand{\ee}{\end{equation}}
\newcommand{\qu}{\quad}
\renewcommand{\t}{\tau}
\newcommand{\tb}{\bar{\tau}}
\newcommand{\pa}{\partial}
\newcommand{\tpt}{\tau\frac{\partial}{\partial\tau}}
\newcommand{\btpt}{\bar{\tau}\frac{\partial}{\partial\bar{\tau}}}

\newtheorem{Theorem}{Theorem}[section]

\newtheorem{Corollary}[Theorem]{Corollary}

\newtheorem{Lemma}[Theorem]{Lemma}

\begin{document}
\begin{small}
\title{Non-chiral 2d CFT with integer energy levels}
\author{
M. Ashrafi\thanks{E-mail address: maryam.ashrafi@ph.iut.ac.ir} ~ and F. Loran\thanks{ E-mail address: loran@cc.iut.ac.ir}\\[6pt]
Department of Physics, Isfahan University of Technology, \\ Isfahan 84156-83111, Iran}
\date{}
\maketitle
\begin{abstract}
The partition function of 2d conformal field theory is a modular invariant function. It is known that the partition function of a holomorphic  CFT whose central charge is a multiple of 24 is a polynomial in the Klein function. In this paper, by using the medium temperature expansion we show that every modular invariant partition function can be mapped to a holomorphic partition function whose structure can be determined  similarly.  We use this map to study  partition function of  CFTs with half-integer left and right conformal weights. We show that the corresponding left and right central charges are necessarily  multiples of 4. Furthermore, the degree of degeneracy  of high-energy levels can be uniquely determined in terms of the degeneracy in the low energy states.
\end{abstract}
\end {small}

\section{Introduction}
   An important question in conformal field theory (CFT) is to what extent a theory can be identified in terms of its constraints and symmetries. The bootstrap hypothesis  \cite{polyakov1,migdal,polyakov2} is based on  the crossing symmetry. Recently in 4d CFT the crossing symmetry has been used to obtain  an upper bound on the weights of the fields that appear in the operator product expansion of scalar operators \cite{rat1}-\cite{ag} and a lower bound on the stress tensor central charge \cite{rat4, Poland:2010wg}. Similarly  an upper bound on the scaling dimension of the first scalar operator appearing in the OPE of two quasi-primary scalar operators has been  obtained in two dimensions \cite{rat1}.

     In two dimensions, the infinite dimensional group of the conformal symmetry makes the bootstrap project more efficient. Furthermore, the partition function of a 2d CFT should be  invariant under modular transformations. The  modular group $\mathbf{PSL}(2,\mathbb{Z})$ is the disconnected diffeomorphism group of the torus
    \begin{equation}
    (\tau,\tb)\rightarrow(\tau^\prime,\tb')=\left(\frac{a\t+b}{c\t+d},\frac{a\tb+b}{c\tb+d}\right),\quad\quad\left(\begin{array}{cc} a & b \\ c & d \\ \end{array} \right)\in\mathbf{PSL}(2,\mathbb{Z}),
    \end{equation}
    where $\t=\t_1+i\t_2$ is the complex structure taking value in the upper half plane ($\t_2\ge 0$) and $\tb=\t_1-i\t_2$.
    The generators of the modular group are
    \begin{align}
    &T:\ (\tau,\tb)\rightarrow(\tau +1,\tb+1), &S:\ (\tau,\tb)\rightarrow \left(-\frac{1}{\tau}, -\frac{1}{\bar\tau}\right).
    \end{align}
    Invariance under  $T$-transformation (henceforth $T$-invariance) constrains the spin of states and the difference between the left and right central charges of the conformal field theory. $S$-invariance constrains the density of states and the spectrum of the theory.

    In  \cite{cardy}  the $S$-invariance of partition function has been used  to estimate the density of states in the saddle-point approximation for a unitary CFT. It is seen that the density of states at conformal dimension $h$ grows exponentially with the square  root of $h$ \cite{carlip}.  The Cardy formula is a key ingredient in the AdS$_3$/CFT$_2$ correspondence; it reproduces the Bekenstein-Hawking entropy of the BTZ-black holes \cite{Strominger:1997eq}.\footnote{The asymptotic symmetry group of an  asymptotically AdS$_3$ spacetime with radius $\ell$ in Planck units is given by two copies of the Virasoro algebra whose central charge $\sim \ell$ \cite{Brown-Henneaux}. The validity of the semi-classical gravity requires that $\ell\gg1$.} $S$-invariance has been also used to compute the `logarithmic correction' \cite{carlip} and `the beyond the logarithmic corrections' to the Cardy formula \cite{loran}. In  \cite{Hartman:2014oaa} it is shown that in theories with sparse light spectrum and large central charge, the Cardy formula also works for energies greater than the central charge.

    Recently, the modular invariance of  partition function  has been used in order to obtain an upper bound on conformal dimensions of the primary fields. In \cite{witten} for  holomorphically factorizable models  whose left and right central charges  are  multiples of 24, an upper bound on the lowest primary fields has been obtained,
    \be
    \Delta\leq \min\left(\frac{c_L}{24}+1, \frac{c_R}{24}+1\right).
    \ee
   This upper bound is  saturated in extremal CFTs \cite{hon1, hon2}. Extremal CFT's are promising holographic duals to the  pure gravity with negative cosmological  constant \cite{witten,Maloney:2007ud}. The vacuum state corresponds to the AdS space, and the primary fields above the vacuum correspond to the BTZ black hole.
    Modular invariance is  enough to determine the partition function of an extremal CFT.  For $c=24$ an extremal CFT is known  and its uniqueness has been conjectured \cite{flm1, flm2}. The holomorphic and anti-holomorphic parts of the partition function are modular functions. A modular function can be written in terms of a polynomial in the Klein function $J$ \cite{appostal}.  While  for the other values of the central charge the partition functions are  known  it is not clear whether  such CFTs exist \cite{gaber1}.

    In general, for CFT's in which there is no chiral algebra beyond the Virasoro algebra and $c\gg1$, the following upper bound on the lowest primary operator  has been obtained \cite{heller1}-\cite{ keller2}
    \begin{align}
    \label{i1}
 &    \Delta \leq \frac{c_{\rm tot}}{12}+{\cal O}(1),  &c_{\rm tot}:=c_L+c_R.
    \end{align}
    In fact for asymptotically large central charge this inequality  is valid for  $\Delta _n$  with $n\leq e^{\frac{\pi c}{12}}$  \cite{qual,quall4}.    The upper bound (\ref{i1}) can be computed by using the medium temperature expansion. This method uses the  $S$-invariance of the partition function  at the self-dual point $\t=-\tb=i$ \cite{heller1}.  Considering a small neighborhood of $\tau=-\tb=i$
    \begin{align}
    \label{5}
 &    \tau= i\, e^s, &\bar\tau= -i\, e^s,
    \end{align}
    in the limit  $s\rightarrow 0$, one obtains  an infinite set of constraints on the partition function:
    \begin{equation}
    \label{6}
    \left.\left(\tau \frac{\partial}{\partial \tau}\right)^{N_{R}} \left(\bar{\tau}\frac{\partial}{\partial \bar{\tau}}\right)^{N_{L}} Z(\tau , \bar{\tau}) \right|_{\tau =-\bar\tau=i} = 0 \quad\quad \mbox{ for} \quad N_L+N_R=\mbox{odd}.
    \end{equation}
    Combining the constraints that can be obtained by different selections of $(N_L,N_R)$ leads to  certain  universal constraints on the spectrum \cite{heller1}, \cite{qual}-\cite{benjamin}.

     In this work we  use the medium temperature expansion method in a different manner. We note that Eq.\eqref{6} indicates that for any (smooth) odd function $f(x,y)=-f(-x,-y)$
     \be
     \left.f\left(\tau \frac{\partial}{\partial \tau},\bar{\tau}\frac{\partial}{\partial \bar{\tau}}\right)Z(\tau , \bar{\tau})\right|_{\tau=-\bar\tau=i} =0.
     \ee
     This observation leads to an interesting result: \textit{corresponding to every $S$-invariant non-chiral partition function $Z(\tau , \bar{\tau})$, there exist an $S$-invariant chiral function ${\cal Z}(\tau)$. }  The map $Z(\tau , \bar{\tau}) \to {\cal Z}(\tau)$ can be interpreted as the {\em chiralization} of the partition function. The chiral function corresponding to the non-chiral partition  function can be easily obtain by inserting $\tb=-\t$ in $Z(\tau , \bar{\tau})$. That is,
     \be
     {\cal Z}(\tau) :=Z(\tau , -\tau).
     \ee
  This equation implies that ${\cal Z}(\tau)$ can be obtained by analytic continuation of the `canonical' partition function
    \be
    Z_{\rm canonical}(\beta):=Z(\tau,\bar\tau)\vert_{\tau=-\bar\tau=\frac{i\beta}{2\pi}}.
    \label{canonical}
    \ee
    to the complex $\beta$-plane. The behavior of the chiral function ${\cal Z}(\tau)$ under $T$ transformation depends on the spectrum of the main theory.

     Focusing on a  special  class of CFTs whose primary operators have half integer scaling dimensions (henceforth HI-CFT), we show  that the corresponding chiral partition function  is an eigen-function of $T$ whose eigen-value is $e^{\frac{-i\pi c_{\rm tot}}{12}}$.
      \begin{align}
      &T: {\cal Z}(\tau)\to e^{\frac{-i\pi c_{\rm tot}}{12}}{\cal Z}(\tau).
       \end{align}
       Since ${\cal Z}(\tau)$ is  by construction $S$-invariant,  the identity $(ST)^3=1$ implies that $c_{\rm tot}\in 8\,{\mathbb Z}$. Thus, in such theories  $c_L$ and $c_R$ are inevitably  multiples of 4. We show that the corresponding chiral partition function ${\cal Z}(\tau)$  can be determined in terms of $1+[\frac{k}{3}]$ positive integers.
       \begin{align}
      & {\cal Z}(\tau)=J^{k/3}\sum_{r=0}^{[\frac{k}{3}]}n_r J^{-r},& n_r\in{\mathbb N}.
       \label{degen}
       \end{align}
 Since the degree of degeneracy  of levels in $Z(\tau , \bar\tau)$ and ${\cal Z}(\tau) $ are equivalent (as can be inferred from Eq.\eqref{canonical}), Eq.\eqref{degen} implies that the degree of degeneracy  of  high-energy levels in $Z(\tau , \bar\tau)$ can be uniquely determined in terms of the degeneracy in the low energy states.

    The organization of the paper is  as follows. In sections \ref{sec2}  we review the effect of two constraints on the partition function. One of them is the $T$-invariance and the other one is the simple fact that partition function should be real-valued. In section \ref{sec3} we study the $S$-invariance of partition function and use the medium temperature expansion to obtain the chiralization map. Sections  \ref{sec4} and \ref{sec5} are devoted to the HI-CFT's. We study the chiral partition function ${\cal Z}(\tau)$ in section  \ref{sec4}, and  identify a subclass of HI-CFT partition functions in terms of free-fermions in section  \ref{sec5}. Some technical details are relegated to the appendices. Our main results are summarized in section \ref{Summary}.



    \section{Constraints on the spin values}\label{sec2}
    Consider a two dimensional unitary CFT on a circle of length  $2\pi$. The partition function of the theory at temperature $\frac{1}{\beta}$ and chemical potential $\mu_c$ is as follows
    \be\label{z2}
    Z(\beta , \mu_c):={\rm Tr}\,e^{-\beta H+i\mu P }=\sum_{\Delta,j}\rho(\Delta,j){e}^{-\beta\left(\Delta-\frac{c_{\rm tot}}{24}\right)}e^{i\mu\left(j-\frac{c_{\rm dif}}{24}\right)},
    \ee
    in which $\mu:=\mu_c\beta$, $H$ is the Hamiltonian and $P$ is the momentum on the compact spatial direction.  The eigenvalues of $H$ and $P$ are $\Delta-\frac{c_{\rm tot}}{24}$  and $j-\frac{c_{\rm dif}}{24}$ respectively \cite{fran}.
      \begin{align}
      &c_{\rm tot}:=c_L+c_R, &c_{\rm dif}:=c_L-c_R,
      \end{align}
 where       $c_L$ and $c_R$ are the left and right central charges.
 This partition function can be interpreted as a CFT partition function on a torus whose complex structure is given by
    \begin{align}
    &\t :=\frac{\mu+i\beta }{2\pi},
    &\tb :=\frac{\mu-i\beta}{2\pi}.
    \end{align}
    In this picture, the conformal weights are given by
    \begin{align}
     &h:=\frac{1}{2}(\Delta +j), &\bar{h}:=\frac{1}{2}(\Delta -j),
    \end{align}
    and the partition function can be written as follows.\footnote{In a unitary CFT $h\ge 0$ and $\bar h\ge 0$. Therefore, $-\Delta\le j\le\Delta$.}
    \begin{equation}
    \label{z1}
    Z(\tau,\bar{\tau})=q^{\frac{-c_L}{24}}{\bar q}^{\frac{-c_R}{24}}\sum_{h,\bar{h}=0}\rho(h,\bar{h})\,q^{h}{\bar q}^{\bar{h}},
    \end{equation}
    in which
    \begin{align}
    & q:=    e^{2 i\pi\tau},&\bar q:=e^{-2 i\pi\bar\tau}.
    \end{align}
     Henceforward we assume the following:
    \begin{itemize}
    \item The partition function is invariant under modular transformation;
    \item The spectrum  contains the identity operator $h=\bar h=0$;
    \item The density of states $\rho(h,\bar{h})$ are positive integer numbers;
    \item The partition function is real.
    \end{itemize}
     In the following,  we show that $T$-invariance  indicates that  the spin $j\in\mathbb{Z}$ and $c_{\rm dif}\in24\mathbb{Z}$. Furthermore we show that since the partition function is real-valued, at each energy level, the number of states with spin $j$ and $-j+\frac{c_{\rm dif}}{12}$ are equivalent.

    \subsection{$T$-invariance of partition function}\label{I}
    \indent  Under $T$ transformation
    \be\label{tt1}
     \mu\rightarrow \mu+2\pi,\quad\quad \beta\rightarrow\beta.
     \ee
      Therefore, $T$-invariance of the partition function  requires that
    \be\label{zt1}
    \sum_{\Delta, j}\rho(\Delta , j){e}^{-\beta\left(\Delta-\frac{c_{\rm tot}}{24}\right)} e^{i\mu\left(j-\frac{c_{\rm dif}}{24}\right)}=\sum_{\Delta , j} \rho(\Delta , j) {e}^{-\beta\left(\Delta-\frac{c_{\rm tot}}{24}\right)}e^{i(\mu+2\pi)\left(j-\frac{c_{\rm dif}}{24}\right)}.
    \ee
    For $\mu=0$ Eq.\eqref{zt1} gives
    \be\label{zt2}
    \sum_{\Delta, j}\rho(\Delta , j){e}^{-\beta(\Delta-\frac{c_{\rm tot}}{24})}\left[1-\cos2\pi \Big(j-\frac{c_{\rm dif}}{24}\Big)\right]=0,
    \ee
    \be\label{zt3}
    \sum_{\Delta , j}\rho(\Delta , j){e}^{-\beta(\Delta-\frac{c_{\rm tot}}{24})}\sin2\pi \left(j-\frac{c_{\rm dif}}{24}\right)=0.
    \ee
    The summands in   (\ref{zt2}) are non-negative. Consequently $j-\frac{c_{\rm dif}}{24}$ is necessarily an integer.  The   vacuum state ($j=0$) enforces that $c_{\rm dif}\in 24\mathbb{Z}$. Therefore,  $j\in\mathbb{Z}$. From Eq.(\ref{zt1}) one  verifies that  these conditions are also sufficient.

    \subsection{Partition function is real-valued}\label{II}
   The imaginary  part of the partition function (\ref{z2}) is zero.
    \be\label{r3}
    \sum_\Delta\sum_{ j\in{\cal J}_\Delta}\rho(\Delta , j) {e}^{-\beta \left(\Delta -\frac{c_{\rm tot}}{24}\right)} \sin\left[ \mu\left( j-\frac{c_{\rm dif}}{24}\right)\right]=0,
    \ee
   where ${\cal J}_\Delta\subset[-\Delta,\Delta]$ denotes the set of spins of states with energy  $\Delta$.
  From  the $T$-invariance  we already  know that $j-\frac{c_{\rm dif}}{24} \in \mathbb{Z}$. Using the orthogonality of $\sin\left[ \left(j-\frac{c_{\rm dif}}{24} \right)\mu\right]$ (as a function of $\mu$) in Eq.\eqref{r3} one obtains
    \begin{align}
    \label{r8}
    &\sum_{\Delta}\left[\rho(\Delta,j)-\rho\left(\Delta,-j +\frac{c_{\rm dif}}{12} \right)\right] {e}^{-\beta \left(\Delta -\frac{c_{\rm tot}}{24}\right)}=0.
    \end{align}
   Assuming the ordering  $\Delta_1<\Delta_2<\cdots$, Eq.\eqref{r8} reads
    \be\label{r9}
 \rho(\Delta _1,j)-\rho\left(\Delta _1,-j+\frac{c_{\rm dif}}{12}\right)+\sum_{\Delta =\Delta_2}\left[\rho(\Delta ,j)-\rho\left(\Delta ,-j+\frac{c_{\rm dif}}{12}\right)\right]{e}^{-\beta\left(\Delta -\Delta _1\right)}=0.
    \ee
    By considering the  $\beta\rightarrow\infty$ limit one verifies that
    \be\label{r12}
    \rho(\Delta _1,j)=\rho\left(\Delta _1,-j+\frac{c_{\rm dif}}{12}\right).
    \ee
    Using Eq.\eqref{r12} in Eq.\eqref{r9}, the same argument implies that $\rho(\Delta _2,j)=\rho\left(\Delta _2,-j+\frac{c_{\rm dif}}{12}\right)$. Iteration gives,
    \be\label{r11}
    \rho(\Delta _m , j)=\rho\left(\Delta _m ,-j+\frac{c_{\rm dif}}{12}\right).
    \ee
    Since $j-\frac{c_{\rm dif}}{24}$ is the momentum eigen-value, we conclude that
     \begin{Corollary}\label{Cr}
  A 2d CFT  whose partition function is real-valued and $T$-invariant is  parity even.
  \end{Corollary}
    \section{Invariance of partition function under $S$-transformation}\label{sec3}
    $S$-invariance of the partition function,
    \be
    Z(\t , \tb)=Z\left(-\frac{1}{\t}, -\frac{1}{\tb}\right),
    \ee
    implies that  \cite{heller1},\footnote{We assume that the partition function is a smooth function of $\beta$ and $\mu$ i.e. there are no phase transitions.}
     \be
    \left(\tpt\right)^{N_L}\left(\btpt\right)^{N_R}Z(\t , \tb)=(-1)^{N_L+N_R}\left(\omega \frac{\pa}{\pa\omega }\right)^{N_L} \left(\bar{\omega}\frac{\pa}{\pa\bar{\omega}}\right)^{N_R}Z\left(\omega, \bar\omega\right),
    \ee
    in which
    \begin{align}
    &\omega := -\frac{1}{\t},&\bar\omega:=-\frac{1}{\bar\tau}.
    \end{align}
     At the self dual point $\t=\omega=i,$ and $ \tb =\bar\omega=-i$, this condition reads
    \be\label{de1}
    {\hat D_L}^{N_L}{\hat D_R}^{N_R}Z(\t , \tb)\bigg\vert _{\tau =+i,\tb =-i}=0 \qu \mbox{for}\qu N_L+N_R=\mbox{odd},
    \ee
    where ${\hat D_L}=\tpt$ and ${\hat D_R}=\btpt$ are respectively the left and the right dilatation operators. For a holomorphic test function ${\cal F}(\tau)$
    \begin{align}
    \label{de2}
 &    e^{x\hat{D}}{\cal F}(\t)={\cal F}(e^x\t),& x\in {\mathbb C}.
    \end{align}
 Eq.\eqref{de1} implies that for any (smooth) odd function $f(-x_L,-x_R)=-f(x_L,x_R)$
    \be\label{di1}
    f({\hat D_L}, {\hat D_R})Z(\t, \tb)\bigg\vert_{\t =i ,\tb =-i}=0.
    \ee
    Using $f_1({\hat D_L}, {\hat D_R})=\sinh(x_L{\hat D_L})\cosh(x_R{\hat D_R})$ and $f_2({\hat D_L}, {\hat D_R})=\cosh(x_L{\hat D_L})\sinh(x_R{\hat D_R})$  and for $x_L, x_R \in \mathbb{C}$ one verifies that
    \be\label{c11}
    Z(u_L , u_R)+Z\left(u_L,\frac{-1}{u_R}\right)-Z\left(\frac{-1}{u_L} , u_R\right)-Z\left(\frac{-1}{u_L},\frac{-1}{u_R}\right)=0,
    \ee
    \be\label{c12}
    Z(u_L , u_R)-Z\left(u_L,\frac{-1}{u_R}\right)+Z\left(\frac{-1}{u_L} , u_R\right)-Z\left(\frac{-1}{u_L},\frac{-1}{u_R}\right)=0,
    \ee
    where $u_L=ie^{x_L}$ and $u_R=-ie^{ix_R}$.\footnote{Since the growth of $\rho(h,\bar{h})$ in Eq.\eqref{z1} is controlled by the Cardy formula, the partition function
    \be
    Z(u_L,u_R)=e^{-\frac{\pi iu_Lc_L}{12}}e^{\frac{\pi i u_R c_R}{24}}\sum_{h,\bar{h}=0}\rho(h,\bar{h})e^{2\pi i u_L h}{e}^{-2\pi i u_R\bar{h}},
    \ee
 is convergent if the imaginary parts of $u_L$ and $u_R$ are positive and negative respectively. We assume that $Z(u_L,u_R$) gives a  biholomorphic analytic continuation of $Z(\beta,\mu)$ to complex $u_L$ in the upper half-plane and $u_R$ in the lower half-plane. }  An immediate result of the identities \eqref{c11} and \eqref{c12} is
    \begin{Corollary}
    Every $S$-invariant partition function $Z(\tau,\bar\tau)$ is extended $S$-invariant, i.e.
    \begin{align}\label{c3}
   & Z(u_L , u_R)=Z\left(\frac{-1}{u_L},\frac{-1}{u_R}\right),
    \end{align}
    where $u_L$ and $u_R$ are two independent $\mathbb C$ parameters taking value  in the upper half-plane and in the lower half-plane respectively.
    \end{Corollary}


    \subsection{Chiralization of the  partition Function }
    Consider the case $u_L=-u_R=\t$ and define
    \be
    {\cal Z}(\tau):=Z(\tau,-\tau).
    \ee
    Eq.\eqref{z1} (for $q=\bar q$) gives
    \be\label{z3}
    {\cal Z}(\tau)=q^{-\frac{c_{\rm tot}}{24}}\sum_{\Delta=0}\hat\rho(\Delta)q^\Delta,
    \ee
    where\footnote{${\cal Z}(\tau)$ corresponds to the analytic continuation of the canonical partition function $Z_{\rm canonical}(\beta)$ defined in Eq.\eqref{canonical} to the complex $\beta$-plane. It is  known that $Z_{\rm canonical}(\beta)$ is a real analytic function \cite{keller}, thus ${\cal Z}(\tau)$ is well-defined. The $S$-invariance of ${\cal Z}(\tau)$ is also indicated by the $S$-invariance of  $Z_{\rm canonical}(\beta)$.}
    \be
    \hat \rho(\Delta):=\sum_{j\in {\cal J}_\Delta}\rho(\Delta,j).
    \ee
    From (\ref{c3}) we learn that the function ${\cal Z}({\t})$ is invariant under $S$-transformation.  In summary,
    \begin{Corollary}\label{C1}
    Corresponding to every $S$-invariant  partition function $Z(\tau,\bar\tau)$, there is a an  $S$-invariant chiral function ${\cal Z}(\tau):=Z(\tau,-\tau)$.
    \end{Corollary}

    We call the map
    \be
    \label{map}
    \mathit{ch} : Z(\t , \tb)\rightarrow{\cal Z}({\t} ),
    \ee
    the chiralization map  and  ${\cal Z}({\t} )$ the $\mathit{ch}$-image of $Z(\t , \tb)$.
     Table \ref{T1} shows some example of the known partition function and the corresponding $\mathit{ch}$-images.
    \begin{table}[htp]

            \renewcommand*{\arraystretch}{2}
            \begin{tabular}{| c | c |c| }
                \hline
                { Model}  &  $Z(\t,\tb)$ & { $\mathit{ch}$-image}\\
                \hline
                { Ising model} & $\frac{1}{2}\left(\bigg\vert\frac{\theta_2}{\eta}\bigg\vert+\bigg\vert\frac{\theta_3}{\eta}\bigg\vert+\bigg\vert\frac{\theta_4}{\eta}\bigg\vert\right)$ &$\frac{1}{2}\left(\frac{\theta_2}{\eta}+\frac{\theta_3}{\eta}+\frac{\theta_4}{\eta}\right)$\\
                \hline
               {Free boson} & $\frac{1}{\sqrt{\t _i}}\frac{1}{\vert\eta(\t)\vert ^2}$& $\frac{1}{\sqrt{-i\t}}\frac{1}{(\eta(\t)) ^2}$\\
                \hline
               { Free boson on a circle of $r=1$} & $\frac{1}{2}\left(\bigg\vert\frac{\theta_2}{\eta}\bigg\vert^2+\bigg\vert\frac{\theta_3}{\eta}\bigg\vert^2+\bigg\vert\frac{\theta_4}{\eta}\bigg\vert^2\right)$&$\frac{1}{2}\left[\left(\frac{\theta_2}{\eta}\right)^2+\left(\frac{\theta_3}{\eta}\right)^2+\left(\frac{\theta_4}{\eta}\right)^2\right]$\\
                \hline
            \end{tabular}
            \centering \caption{\label{T1} \footnotesize  Examples of non-chiral partition functions and the corresponding $\mathit{ch}$-images.}
        \end{table}
        %
    \section{CFT's with half-integer conformal weights}\label{sec4}
    In this section we investigate a family of CFT's in which $\Delta\in {\mathbb Z}$. Since $j\in {\mathbb Z}$, the corresponding conformal weights are half-integers. Hence we call such a CFT an {\em HI-CFT}. In the following $Z(\tau,\bar\tau)$ and ${\cal Z}(\tau)$ denote the partition function of an HI-CFT and the corresponding $\mathit{ch}$-image respectively.

    From Eq.\eqref{z3} one verifies that
    \be
    T:\ {\cal Z}(\tau)\to e^{-i\pi\frac{c_{\rm tot}}{12}}{\cal Z}(\tau).
    \label{T-phase}
    \ee
    Since ${\cal Z}(\tau)$ is $S$-invariant, using the identity $(TS)^3=1$  one obtains
    \be\label{st1}
    e^{-2\pi  i\frac{c_{\rm tot}}{8}}=1.
    \ee
    Consequently,
    \be\label{c1}
     c_{\rm tot}\in 8 \,{\mathbb N}.
     \ee
     From the $T$-invariance of $Z(\tau,\bar\tau)$ we have learned that $c_{\rm dif}\in 24\,\mathbb{Z}$. Therefore,
     \begin{Corollary}\label{C2}
     For an HI-CFT
      \begin{align}
      \label{c2}
 &      c_L\in\,4\,{\mathbb N}, & c_R\in 4\,{\mathbb N}.
      \end{align}
      \end{Corollary}
     Now we are ready to obtain the basis for ${\cal Z}({\t} )$. Let's start with  $c_L, c_R\in\,12 \mathbb{N}$. In that case $c_{\rm tot}\in 24{\mathbb Z}$ and ${\cal Z}({\t})$ defined by Eq.\eqref{z3} is a well-defined modular invariant meromorphic function in the upper half plane. Therefore it can be given as a  polynomial in the Klein function  $J$ \cite{appostal},
     \be
     {\cal Z}=\sum_{r=-\frac{c_{\rm tot}}{24}}^0 a_r J^{-r}.
     \ee
      The Klein function can be written  in terms of the Jacobi Theta functions $\theta_i(\t)$ $ (i=2,3,4)$ and the Dedekind  function $\eta(\tau)$.
    \bea
    \label{g1}
    J&=&\mathfrak{j}^3\\
    \label{j1}
    &=&q^{-1}+744+196884\, q+\cdots,
    \eea
    where
    \bea
    \label{g11}
    \mathfrak{j} (\t)&:=&\frac{1}{2}\left[ \left({\frac{\theta_2(\t)}{\eta(\t)}}\right)^{8}+\left({\frac{\theta_3(\t)}{\eta(\t)}}\right)^{8}+\left({\frac{\theta_4(\t)}{\eta(\t)}}\right)^{8}\right] \nn\\
     &=&q^{\frac{-1}{3}}\left(1+248\,q+\cdots\right).
     \eea
   In the following we show that for $c_{\rm tot}\in 8{\mathbb N}$, ${\cal Z}(\tau)$  can be written in terms of a polynomial in $\mathfrak{j}$.
    \begin{Lemma}\label{t41}
    Let $f^{(r)}\left(\lbrace a^{(r)}\rbrace,\t\right)$ be an $S$-invariant function with  Fourier expansion
    \begin{align}\label{g31}
    &f^{(r)}\left(\lbrace a^{(r)}\rbrace,\t\right)=q^{\frac{-p}{3}}\left[\sum _{n=-r}^{0} a_n^{(r)} q^n+\sum _{n=1}^{\infty } a_n^{(r)} q^{n}\right],&p\in\{0,1,2\}.
    \end{align}
    in the upper half  $\tau$-plane.  Then
    \begin{itemize}
    \item[a.]  $f^{(r)}\left(\lbrace a^{(r)}\rbrace,\t\right)$  is $T^3$-invariant.
    \item[b.] It is a polynomial in $\mathfrak{j}$.
    \end{itemize}
    \end{Lemma}
    \begin{proof}
    $T^3$-invariance is obvious. Eq.\eqref{j1} and Eq.\eqref{g11} imply that there exist $\lbrace a^{(r-1)}\rbrace$ such that
    \be\label{g311}
    q^{\frac{-p}{3}}\sum _{n=-r}^{\infty} a_n^{(r)} q^n=a_{-r}^{(r)}\,\mathfrak{j}^p\, J^r+q^{\frac{-p}{3}}\sum _{n=-r+1}^{\infty} a_n^{(r-1)} q^n.
    \ee
    Therefore,
    \be\label{g312}
    f^{(r)}\left(\lbrace a^{(r)}\rbrace,\t\right)=a_{r}^{(r)}\,\mathfrak{j}^p\, J^{r}+f^{(r-1)}\left(\lbrace a^{(r-1)}\rbrace,\t\right).
    \ee
   The order of the poles of $f^{(r)}\left(\lbrace a^{(r)}\rbrace,\t\right)$ and $f^{(r-1)}\left(\lbrace a^{(r-1)}\rbrace,\t\right)$ are  $r$ and $r-1$ respectively.
   The $S$-invariance of $f^{(r)}\left(\lbrace a^{(r)}\rbrace,\t\right)$, $\mathfrak{j}$ and $J$ imply that  $f^{(r-1)}\left(\lbrace   a^{(r-1)}\rbrace,\t\right)$ is also $S$-invariant. By iteration one obtains
    \be\label{g313}
    f^{(r)}\left(\lbrace a^{(r)}\rbrace,\t\right)=\mathfrak{j}^p\left[a_{-r}^{(r)}J^{r}+a_{-(r-1)}^{(r-1)}J^{r-1}+\cdots +a_{0}^{(0)}\right]+f^{(-1)}\left(\lbrace a^{(-1)}\rbrace,\t\right),
    \ee
    where
    \be\label{g32}
    f^{(-1)}\left(\lbrace a^{(-1)}\rbrace,\t\right)=q^{1-\frac{p}{3}}\sum_{m\ge1} a^{(-1)}_m q^{m-1}.
    \ee
    The function $\left[f^{(-1)}\left(\lbrace a^{(-1)}\rbrace,\t\right)\right]^3$ is modular invariant. It has no pole in the upper half plane and is zero at     $\tau=i\infty$. Thus it is zero in the upper half plane.
    \end{proof}
    \begin{Corollary}\label{C3}
    The  $\mathit{ch}$-image of the  HI-CFT partition function with total central charge $c_{tot}=8k$, has an expansion in terms of      $\mathfrak{j} $ as follows
     \begin{align}
     \label{ch1}
    &{\cal Z}(\t)=\mathfrak{j}^{k}\sum_{r=0}^{[k/3]}n_r J^{-r}, &n_r\in {\mathbb N}.
    \end{align}
    \end{Corollary}
   The degeneracy of the vacuum state is given by $n_0$. In the following we assume that $n_0=1$.  Eq.\eqref{ch1} shows that ${\cal Z}({\t} )$, and consequently the number of states  with energy $\Delta$ i.e. $\hat \rho(\Delta)$ can be uniquely determined if $c_{\rm tot}$ and the integers $n_r$, or equivalently, the low-energy (i.e. $\Delta\leq \left[\frac{k}{3}\right]$) density of states are given.\footnote{ The existence and the uniqueness of such CFT's is an open problem.}

    Finally, consider an HI-CFT whose $\mathit{ch}$-image  ${\cal Z}(\t)$ is extremal, i.e. ${\cal Z}(\t)=q^{-k/3}\left[1+{\cal O}(q)\right]$. In that case,  the coefficients $n_r$ can be uniquely determined in terms of the central charge. Furthermore, the scaling dimension of the first primary field after identity is $\Delta_1=\frac{c_{\rm tot}}{24}+1$, which is in agreement with the upper bound given in Eq.\eqref{i1}.

    \subsection{AdS/CFT correspondence}
    It is known that the Cardy formula reproduces the Bekenstein-Hawking entropy at $\Delta\gg1$. In \cite{witten} it has been observed that for $k\in3\,{\mathbb N}$,\footnote{In our conventions, $c_{\rm tot}= 8k$ while in \cite{witten} $c_{\rm tot}=24 k$.} the  number of primary fields is given by the Cardy formula
    \begin{align}
    &\hat\rho(\Delta)\cong e^{{\mathcal{ S}}(k,\Delta)},&\mathcal{S}(k,\Delta):=4\pi\sqrt{\frac{k\left(\Delta-\frac{k}{3}\right)}{3}}.
        \end{align}
    Therefore it is natural to assume that the  primary fields correspond to the micro-states of the BTZ black hole.

    In Table \ref{tab2} the Fourier expansion of the $\mathit{ch}$-image is given for $c_{\rm tot}=8, 16,24$. The coefficients of the expansions determine the density of state $\hat \rho(\Delta)$ which equals the number of states with energy $\Delta$ and spin $j\in[-\Delta,\Delta]$.
    \begin{table}[htp]
                \renewcommand*{\arraystretch}{2}
            \begin{tabular}{| c | c |c| }
            \hline
            central charge & $\mathit{ch}$-Image of partition function\\
                \hline
                $k=1$  &$ \mathfrak{j}=q^{\frac{-1}{3} } (1+248q+4124q^2+34752q^3+\cdots)$\\
                \hline
               $k=2$&  $\mathfrak{j}^2=q^{\frac{-2}{3} } (1+496q+69752q^2+2115008q^3+\cdots)$ \\
                \hline
               $k=3$& $ J+n_1=q^{-1} \left[1+(744+n_1)q+196884q^2+21493760q^3+\cdots\right]$\\
               \hline
               \end{tabular}

                \centering \caption{\label{tab2} \footnotesize The $\mathit{ch}$-image of the HI-CFT partition function with $c_{\rm tot}\in\{8,16,24\}$.}
        \end{table}
        For $k=1,2,3$ the first high energy state (i.e.  $\Delta=1+\left[\frac{k}{3}\right]$) has weight $\Delta=1, 1, 2$ respectively. It is an interesting observation that the corresponding number of states can be estimated by the Cardy formula.

    \section{A basis for   HI-CFT partition function}\label{sec5}
    In the previous section we have observed that the $\mathit{ch}$-image of the HI-CFT partition function
    is a polynomial in $\mathfrak{j}$. Motivated by the fact that $\mathfrak{j}$ is the $\mathit{ch}$-image of
    $\frac{1}{2}\sum_{i=1}^3\left|\frac{\theta_i }{\eta} \right|^8$ in this section we study a class of HI-CFT's whose partition functions
    can be given as a polynomial in  $\sqrt{\frac{\theta_i }{\eta}}$  and $\sqrt{\bar{\frac{\theta_i }{\eta}}}$.

    The functions $\sqrt{\frac{\theta_i }{\eta}}$ have the following Fourier expansion.
    \be\label{g140}
    \sqrt{\frac{\vartheta_2}{\eta}}=q^{\frac{-1}{48}+\frac{1}{16}}\sum_{n=0}^{\infty}C^{(\frac{1}{16})}_{n}q^{n},
    \ee
    \be\label{g150}
    \sqrt{\frac{\vartheta_3}{\eta}}=q^{\frac{-1}{48}}\left(\sum _{n=0}^{\infty} C^{(0)}_{n} q^{n} +\sum _{n=1}^{\infty}C^{(\frac{1}{2})}_{n} q^{n+\frac{1}{2}}\right),
    \ee
    \be\label{g160}
    \sqrt{\frac{\vartheta_4}{\eta}}=q^{\frac{-1}{48}}\left(\sum _{n=0}^{\infty} C^{(0)}_{n}q^{n} -\sum _{n=1}^{\infty}C^{(\frac{1}{2})}_{n} q^{n+\frac{1}{2}}\right),
    \ee
    where
    \begin{align}
    &C^{(i)}_n\in{\mathbb N}, &i=0,\frac{1}{16}, \frac{1}{2}.
    \end{align}
    The $S$-transformation of  the Dedekind function $\eta$ and the Theta functions are as follows.
     \begin{align}\label{g3}
    &\theta _2\rightarrow (-i\t)^{1/2} \theta _4,& \theta _4\rightarrow (-i\t)^{1/2} \theta
    _2,\\
    & \theta _3\rightarrow (-i\t)^{1/2} \theta _3,& \eta \rightarrow (-i\t)^{1/2}
    \eta.
    \end{align}
    $T$-transformation of these functions is given by,
    \begin{align}\label{g4}
    &    \theta _2\rightarrow e^{i\pi /4}\theta _2,&&  \theta _3\leftrightarrow \theta _4,
    &    &\eta\rightarrow e^{i\pi /12}\eta.
    \end{align}
    Since an HI-CFT only contains primary fields with half integer scaling dimension, from Eqs.(\ref{g140}-\ref{g4})
    one infers that the corresponding partition function is a  polynomial in $x$, $y$ and $z$ defined as
    follows.
    \be\label{g14}
    x:=\left(\sqrt{\frac{\vartheta_2}{\eta}}\right)^8=q^{\frac{-1}{6}+\frac{1}{2}} C(q),
    \ee
    \be\label{g15}
    y:=\left(\sqrt{\frac{\vartheta_3}{\eta}}\right)^8= q^{\frac{-1}{6}}\left[A(q)+q^{\frac{1}{2}}B(q)\right],
    \ee
    \be\label{g16}
    z:=\left(\sqrt{\frac{\vartheta_4}{\eta}}\right)^8= q^{\frac{-1}{6}}\left[A(q)-q^{\frac{1}{2}}B(q)\right],
    \ee
    where $A(q)$, $B(q)$ and $C(q)$ are polynomials in $q$ with positive integer coefficients.
    The functions $x,y$ and $z$ are not independent. They are related through the standard relations between the Theta
    functions and Dedekind function.
    \be\label{g17}
    x-y+z=0,
    \ee
    \be\label{g170}
     xyz=16.
    \ee
     By using Eq.\eqref{g170}  and  the transformation rules
    \begin{align}
    \label{xt}
    &S:\ y\rightarrow y,&\ x\leftrightarrow z,\\
    \label{xt1}
    &T:\ x\rightarrow e^{2i\pi /3}x,&& y\rightarrow e^{-i\pi /3} z,&&z\rightarrow e^{-i\pi /3} y,
    \end{align}
    one can show that the most general modular covariant combination of $x,y, z$ can be written as
    follows.
    \begin{eqnarray}\label{b130}
    R_{a,b,c,d}&=&x^{c}\bar{x}^{a}\left(y^{d}\bar{z}^{b}+\alpha z^{d}\bar{y}^b\right)\\\nonumber
    &+&y^{c}\bar{y}^{a}\left(\beta z^{d}\bar{x}^{b}+\tilde{\beta} x^{d}\bar{z}^b\right)\\\nonumber
    &+&z^{c}\bar{z}^{a}\left(\gamma x^{d}\bar{y}^{b}+\tilde{\gamma } y^{d}\bar{x}^b\right),
    \end{eqnarray}
    where $a, b, c$ and $d$ are some positive integer number and $\alpha, \beta, \tilde{\beta}, \gamma$ and $\tilde{\gamma}$ are complex number.
    By covariance we mean that
    \be\label{g70}
    S:\ R_{a,b,c,d}\to e^{i\sigma } R_{a,b,c,d},
    \ee
    \be\label{g80}
    T:\ R_{a,b,c,d}\to e^{i\delta } R_{a,b,c,d},
    \ee
     where $\sigma$ and $\delta$ are real numbers. Eq.\eqref{g70} implies that
    \be\label{g7}
    \tilde{\gamma}=e^{i\sigma },\quad\quad\tilde{\gamma}^2=e^{2i\sigma }=1,\quad\quad \gamma =e^{i\sigma }\alpha,\quad\quad \tilde{\beta}=e^{i\sigma }\beta.
    \ee
    Using (\ref{g80}) and \eqref{g7} one obtains,
    \begin{align}\label{g8}
  &  e^{i\delta}=(-1)^{a+c}e^{\frac{i\pi}{3}(a+b-c-d) }\alpha,& \beta= (-1)^{a+c+d},\\&
  \tilde{\gamma}=(-1)^{b+d}\alpha,& \alpha ^2=1.
    \end{align}
   Consequently,
    \begin{eqnarray}\label{b13}
    R^{\pm}_{a,b,c,d}&=&x^{c}\bar{x}^{a}\left((-y)^{d}\bar{z}^{b}\pm z^{d}(-\bar{y})^b\right)\\\nonumber
    &+&(-y)^{c}(-\bar{y})^{a}\left(z^{d}\bar{x}^{b}\pm x^{d}\bar{z}^b\right)\\\nonumber
    &+&z^{c}\bar{z}^{a}\left(x^{d}(-\bar{y})^{b}\pm (-y)^{d}\bar{x}^b\right),
    \end{eqnarray}
    where we have dropped an overall phase $(-1)^d$. $R^{-}_{a,b,c,d}$ and $R^+_{a,b,c,d} $ are respectively odd and even under $S$-transformation.
    \be
    S R^{\pm}_{a,b,c,d}=\pm R^{\pm}_{a,b,c,d}.
    \ee
    $S$-invariance of the partition function implies that $Z(\t,\bar{\t})$ should be an even function in $R^-_{a,b,c,d}$.
     In Appendix \ref{apb} we show that $R^-_{a,b,c,d} R^-_{a',b',c',d'}$ is a linear combination of $R^+_{a,b,c,d}$.
     Hence, we concentrate on polynomials in $R^+_{a,b,c,d}$ and drop the    $+$ sign for simplicity.

     Noting that
    \be\label{g21}
    \mathfrak{j} =\frac{1}{2}(x^2+y^2+z^2),
    \ee
    one can use Eq.\eqref{g17} to show that
    \be\label{g210}
    \mathfrak{j}=x^2+yz=z^2+xy=y^2-xz.
    \ee
    These identities together with  Eq.(\ref{g170}) result in the
    following recurrence relations.
    \be\label{b14}
    R_{a+2,b,c,d}=\bar{\mathfrak{j}}R_{a,b,c,d}-16R_{a-1,b,c,d},
    \ee
    \be\label{b15}
    R_{a,b+2,c,d}=\bar{\mathfrak{j}} R_{a,b,c,d}-16R_{a,b-1,c,d},
    \ee
    \be\label{b17}
    R_{a,b,c+2,d}=\mathfrak{j}R_{a,b,c,d}-16R_{a,b,c-1,d},
    \ee
    \be\label{b18}
    R_{a,b,c,d+2}=\mathfrak{j}R_{a,b,c,d}-16R_{a,b,c,d-1}.
    \ee
        Eqs.\eqref{b14}-\eqref{b18} show that every $R_{a,b,c,d}$  is a polynomial   in $\mathfrak{j}$, $\bar{\mathfrak{j}}$,
         \be\label{b35}
     \mathfrak{h}:=\frac{1}{2}\left(|x|^2+|y|^2+|z|^2\right),
     \ee
     and
     \be\label{b36}
     {\mathfrak{k}}:=x^2\bar{x}-y^2\bar{y}+z^2\bar{z}.
     \ee
      To show it,  we first consider the chiral function
    \be\label{g12}
    R_{c,d}
    :=R_{0,0,c,d}=x^c\left((-y)^{d}+z^{d})+(-y)^c(z^{d}+x^{d})+z^c(x^{d}+(-y)^{d}\right).
    \ee
    Noting that
    \be\label{g23}
    R_{0,0}=6,\quad\quad R_{1,0}=0, \quad\quad R_{2,0}=4\mathfrak{j},
    \ee
    \be\label{g24}
    R_{0,1}=0,\quad\quad R_{1,1}=-2\mathfrak{j}, \quad\quad R_{2,1}=48,
    \ee
    \be\label{g25}
    R_{0,2}=4\mathfrak{j},\quad\quad R_{1,2}=48, \quad\quad R_{2,2}=2\mathfrak{j} ^2,
    \ee
    one verifies that $\mathfrak{j}$ is the single generator of   $R_{c,d}$.

     In appendix \ref{apa} we show that
    \be\label{b19}
    R_{a,b,0,1}=- R_{b,a,1,0}-R_{a,b,1,0},
    \ee
    \be\label{b20}
    R_{a,b,1,1}=-\mathfrak{j}R_{a,b,0,0}+ R_{b,a,2,0},
    \ee
    \be\label{b21}
    R_{a,b,0,2}=2\mathfrak{j} R_{a,b,0,0}-R_{a,b,2,0}- R_{b,a,2,0},
    \ee
    \be\label{b22}
    R_{a,b,2,2}=\mathfrak{j}^2R_{a,b,0,0}-\mathfrak{j} R_{b,a,2,0}-16 R_{b,a,1,0},
    \ee
    \be\label{b24}
    R_{a,b,1,2}=-\mathfrak{j} R_{a,b,0,1}+ R_{b,a,2,1}.
    \ee
    Thus all of $R_{a,b,c,d}$ can be obtained in terms of $R_{a,b,c',d'}$ where $\{c',d'\}\in \{(1,0),(2,0),(2,1)\}$. Using the identity
    \be\label{b25}
    \bar{R}_{a,b,c,d}= R_{c,d,a,b},
    \ee
     one can also determine $R_{2,0,1,0}$, $R_{2,1,1,0}$, and $R_{2,1,2,0}$  in terms of $R_{1,0,2,0}$, $R_{1,0,2,1}$, and $R_{2,0,2,1}$ respectively. Using Eq.(\ref{b20}) one can determine  $R_{1,1,1,0}$, $R_{1,1,2,0}$ and $R_{1,1,2,1}$. Similarly, Eq.(\ref{b22}) can be used to compute  $R_{2,2,1,0}$, $R_{2,2,2,0}$ and $R_{2,2,2,1}$.
    Therefore all that we need to compute $R_{a,b,c,d}$ are the following functions.
    \be
    \begin{array}{ll}
     R_{1,0,1,0}=-2R_{0,1,1,0}=4\mathfrak{h},&
    R_{0,1,2,0}=\bar{ R}_{0,2,1,0}=-\frac{1}{2}R_{1,0,2,0}=-\mathfrak{k},\\
      R_{1,2,1,0}= -\bar{R}_{1,0,2,1} =2\,\bar{\mathfrak{j}}\,\mathfrak{h},&    R_{0,1,2,1}= 0,\\
    R_{2,0,2,0}=\frac{8h^2}{3}+\frac{4|\mathfrak{j}|^2}{3},&
    R_{0,2,2,0}=-\frac{4\mathfrak{h}^2}{3}+\frac{10|\mathfrak{j}|^2}{3},\\
    R_{2,0,2,1}=-\mathfrak{j}\bar{{\mathfrak{k}}}+32\mathfrak{j},&         R_{0,2,2,1}=32\bar{\mathfrak{j}},\\
    R_{2,1,2,1}=\frac{8}{3}\left|\mathfrak{j}\right|^2\mathfrak{h}-\frac{8\mathfrak{h}^3}{3}+6(16)^2.
     \end{array}
     \ee
     In summery every $R_{a,b,c,d}$ is a  polynomial in $\mathfrak{h}, {\mathfrak{k}}$ and $\bar{{\mathfrak{k}}}$ as follows
    \be\label{b42}
   \mathfrak{R} =g_0+g_1\mathfrak{h}+g_2 \mathfrak{h}^2+g_3 \mathfrak{h}^3+g_4{\mathfrak{k}}+g_5\bar{{\mathfrak{k}}},
    \ee
    where $g_i=g_i(\mathfrak{j},\bar{\mathfrak{j}})$ are polynomials in $\mathfrak{j}$ and ${\bar{\mathfrak{j}}}$.

     Since the HI-CFT partition function is modular invariant, we investigate the invariance of $ \mathfrak{R}$ under $T$ transformation. $\mathfrak{h}$ is modular invariant. $\mathfrak{j}$ and ${\mathfrak{k}}$ are eigen-functions of $T$ with eigen-values $e^{\frac{-2\pi i}{3}}$  and $e^{\frac{2\pi i}{3}}$ respectively.
    In order to determine $g_i, (i=0\cdots 5)$, we write them as follows
    \begin{align}\label{gi1-1}
 &   g_i= F_i^{(0)} +\left(F_i^{(1)}\mathfrak{j}+F_i^{(2)}\mathfrak{j}^2+{\rm h.c.}\right),&i=0,1,2,3,\\
  &g_4=F_4^{(0)} +F_4^{(1)}\mathfrak{j}+F_4^{(2)}\mathfrak{j}^2+G_4^{(1)}\bar{\mathfrak{j}}+G_4^{(2)}{\bar{\mathfrak{j}}}^2,\\
  \label{gi1-3}
  &g_5={\bar{g_4}}.
    \end{align}
     where $F_i^{(a)}$ and $G_4^{(a)}$ are polynomials in $ \left|\mathfrak{j}\right|^2$, $J$ and $\bar{J}$.
    In writing Eqs.\eqref{gi1-1}-\eqref{gi1-3} we have  noted that $\mathfrak{R}$ as a partition function should be real-valued. Using Eqs.\eqref{gi1-1}-\eqref{gi1-3} in Eq.(\ref{b42}) and the identity $T \mathfrak{R}+T^2 \mathfrak{R}=2 \mathfrak{R}$ one obtains
    \be\label{gi2}
   \left[\sum_{i=0}^3  \left(\mathfrak{j}F_i^{(1)} +\mathfrak{j}^2 F_i^{(2)}\right) \mathfrak{h}^i+\left(F_4^{(0)}+F_4^{(2)}{\mathfrak{j}}^2+G_4^{(1)}\bar{\mathfrak{j}}\right)\mathfrak{k}\right]+{\rm c.c.}=0.
    \ee
    Therefore, every HI-CFT partition function can be written as
     \be\label{b45}
    Z(\t,\tb)=\sum_{i=0}^3 F_i^{(0)} \mathfrak{h}^i +\left[\left(F_4^{(1)}\mathfrak{j}+G_4^{(2)}{\bar{\mathfrak{j}}}^2\right)\mathfrak{k}+{\rm c.c.}\right].
    \ee
 Noting that the $\mathit{ch}$-image of $\mathfrak{h}$ is $\mathfrak{j}$ and the $ch$-image of $\mathfrak{k}$ equals -48,\footnote{The $\mathit{ch}$-image of $\mathfrak{k}$ is $\frac{1}{2}R_{0,0,3,0}$ which can be easily computed by using Eq.\eqref{b17}.} one easily verifies that the $\mathit{ch}$-image of $Z(\t,\tb)$ is a function of $\mathfrak{j}$ in agreement with corollary \ref{C3}.
     \subsection{Examples of HI-CFT Partition function}
     In this section we study HI-CFT's  with $c_{\rm tot}=8,16$.\footnote{$c_{\rm dif}\in24{\mathbb Z} $ implies that the  corresponding left and right central charges are  $c_L=c_R=4$ and $c_L=c_R=8$ respectively.}
     \begin{itemize}
     \item  $c_{\rm tot}=8$.  In this case there is only one partition function
     \be\label{ex1}
    Z(\t , \tb)=\mathfrak{h}=\frac{1}{2} \frac{{\vartheta_2}^4(\t)\bar{{\vartheta_2}}^4(\tb)+{\vartheta_3}^4(\t)\bar{{\vartheta_3}}^4(\tb)+{\vartheta_4}^4(\t)\bar{{\vartheta_4}}^4(\tb)}{\eta ^4(\t){\bar{\eta}}^4(\tb)},
     \ee
     which corresponds to  8 right-handed  and 8 left-handed fermions. The corresponding $\mathit{ch}$-image is
    \be\label{ex2}
    {\cal Z}({\t})=\frac{1}{2} \frac{{\vartheta_2}^8+{\vartheta_3}^8+{\vartheta_4}^8}{\eta ^8(\t)}=\mathfrak{j}.
    \ee
      \item  $c_{\rm tot}=16$.
      In this case the partition function is not unique.
      \be\label{ex3}
     Z(\t,\tb)=\frac{1}{a+b}\left(a \mathfrak{h}^2+b\left| \mathfrak{j}\right|^2\right).
     \ee
     The $\mathit{ch}$-image of $Z(\t,\tb)$ is $\mathfrak{j}^2$ (independent of $a$ and $b$). The factor $\frac{1}{a+b}$ indicates that there is a single vacuum state. The coefficients $a$ an $b$ should be  determined in such a way that the density of states are positive integers. By inspecting the first few terms in the  Fourier expansion of $Z(\t,\tb)$, one can obtain the following necessary condition.
     \be\label{ex39}
     \frac{384a}{a+b}=n, \quad\quad  \frac{56a+248b}{a+b}=m,
     \ee
  in which $m$ and $n$ are nonnegative integers. This gives
    \begin{align}\label{ex391}
  &   m,n\in8{\mathbb Z}& 2m'+n'=62,
    \end{align}
  where $m':=\frac{m}{8}$ and $n':=\frac{n}{8}$. Using Eq.\eqref{ex391} in Eq.\eqref{ex3} one obtains
     \be\label{ex31}
     Z(\t,\tb)=\frac{1}{24}\left[(31-m') \mathfrak{h}^2+(m'-7)\left|\mathfrak{j}\right|^2\right].
     \ee
    For $7\leq m'\leq 31$ the energy densities are obviously positive integers.  We have not been able to exclude the partition functions corresponding to  $0\le m'\leq 6$. Therefore, we are optimistic that there should be 32 different HI-CFT's with $c_{\rm tot}=16$.

     \end{itemize}
     \section{Summary}\label{Summary}
     In this work we have studied modular invariant partition functions of unitary CFT's whose conformal weights  are half-integers, hence  HI-CFT's. By using the medium temperature expansion we have obtained a chiralization map which maps every  $S$-invariant non-chiral partition function to an $S$-invariant chiral partition function. We have used the chiralization map to show that the left and right central charges of an HI-CFT are multiples of 4. Furthermore, we have shown that the partition function after chiralization can be written as a polynomial in $j=J^{1/3}$, where $J$ is the Klein function. In this way we have realized that the degree of degeneracy of the high energy levels $\Delta>\left[ \frac{c_L+c_R}{24}\right]$ can be uniquely determined in terms $1+\left[\frac{c_L+c_R}{24}\right]$ integers corresponding to the degeneracy in the low energy states.

     We have identified a class of HI-CFT's whose partition functions can be given in terms of the Jacobi Theta function $\theta_i$ and the Dedekind function $\eta$. In Eq.\eqref{b45} we have given  the most general form of such partition functions.
     \subsection*{Acknowledgments}
     We are grateful to M.M. Sheikh-Jabbari for reading the  manuscript and for his useful comments.

     \appendix
      \section{$S$-invariant combinations of $R_{a,b,c,d}$}\label{apb}
     The multiplication rule for  $R^{-}_{a,b,c,d}$ can be obtained as follows.
     \begin{eqnarray}\label{2m}
    R^{-}_{a,b,c,d}R^{-}_{a',b',c',d'}&=&R^{+}_{a+a',b+b',c+c',d+d'}\\\nonumber
    \\\nonumber
    &-&\Big[x^{c+c'}(-y)^{d'}z^{d}\bar{x}^{a+a'}(-\bar{y})^{b}\bar{z}^{b'}+ x^{c+c'}(-y)^{d}z^{d'}\bar{x}^{a+a'}(-\bar{y})^{b'}\bar{z}^{b}\\\nonumber
    &+&x^{d'}(-y)^{c+c'}z^{d}\bar{x}^{b}(-\bar{y})^{a+a'}\bar{z}^{b'}+x^{d}(-y)^{c+c'}z^{d'}\bar{x}^{b'}(-\bar{y})^{a+a'}\bar{z}^{b}\\\nonumber
    &+& x^{d}(-y)^{d'}z^{c+c'}\bar{x}^{b'}(-\bar{y})^{b}\bar{z}^{a+a'}+x^{d'}(-y)^{d}z^{c+c'}\bar{x}^{b}(-\bar{y})^{b'}\bar{z}^{a+a'}\Big]\\\nonumber
    \\\nonumber
    &+&\Big[x^{c}(-y)^{c'+d}z^{d'}\bar{x}^{a+b'}(-\bar{y})^{a'}\bar{z}^{b}+ x^{c}(-y)^{d'}z^{c'+d}\bar{x}^{a+b'}(-\bar{y})^{b}\bar{z}^{a'}\\\nonumber
    &+&x^{d'}(-y)^{c}z^{c'+d}\bar{x}^{b}(-\bar{y})^{a+b'}\bar{z}^{a'}+x^{d'}(-y)^{d+c'}z^{c}\bar{x}^{b}(-\bar{y})^{a'}\bar{z}^{a+b'}\\\nonumber
    &+& x^{c'+d}(-y)^{c}z^{d'}\bar{x}^{a'}(-\bar{y})^{a+b'}\bar{z}^{b}+x^{c'+d}(-y)^{d'}z^{c}\bar{x}^{a'}(-\bar{y})^{b}\bar{z}^{a+b'}\Big]\\\nonumber
    \\\nonumber
    &+&\Big[x^{c'}(-y)^{c+d'}z^{d}\bar{x}^{a'+b}(-\bar{y})^{a}\bar{z}^{b'}+ x^{c'}(-y)^{d}z^{c+d'}\bar{x}^{a'+b}(-\bar{y})^{b'}\bar{z}^{a}\\\nonumber
    &+&x^{d}(-y)^{c'}z^{c+d'}\bar{x}^{b'}(-\bar{y})^{a'+b}\bar{z}^{a}+x^{d}(-y)^{d'+c}z^{c'}\bar{x}^{b'}(-\bar{y})^{a}\bar{z}^{a'+b}\\\nonumber
    &+& x^{c+d'}(-y)^{c'}z^{d'}\bar{x}^{a}(-\bar{y})^{a'+b}\bar{z}^{b'}+x^{c+d'}(-y)^{d}z^{c'}\bar{x}^{a}(-\bar{y})^{b'}\bar{z}^{a'+b}\Big]\\\nonumber
    \\\nonumber
    &-&\Big[x^{c+d'}(-y)^{c'+d}\bar{x}^{a}(-\bar{y})^{a'}\bar{z}^{b+b'}+ x^{c+d'}z^{c'+d}\bar{x}^{a}(-\bar{y})^{b+b'}\bar{z}^{a'}\\\nonumber
    &+&(-y)^{c+d'}z^{c'+d}\bar{x}^{b+b'}(-\bar{y})^{a}\bar{z}^{a'}+(-y)^{d+c'}z^{c+d'}\bar{x}^{b+b'}(-\bar{y})^{a'}\bar{z}^{a}\\\nonumber
    &+& x^{c'+d}(-y)^{c+d'}\bar{x}^{a'}(-\bar{y})^{a}\bar{z}^{b+b'}+x^{c'+d}z^{c'+d}\bar{x}^{a'}(-\bar{y})^{b+b'}\bar{z}^{a}\Big]\\\nonumber
    \\\nonumber
    &-&\Big[x^{c}(-y)^{c'}z^{d+d'}\bar{x}^{a+b'}(-\bar{y})^{a'+b}+ x^{c}(-y)^{d+d'}z^{c'}\bar{x}^{a+b'}\bar{z}^{a'+b}\\\nonumber
    &+&x^{d+d'}(-y)^{c}z^{c'}(-\bar{y})^{a+b'}\bar{z}^{a'+b}+x^{d+d'}(-y)^{c'}z^{c}(-\bar{y})^{a'+b}\bar{z}^{a+b}\\\nn
    &+& x^{c'}(-y)^{c}z^{d+d'}\bar{x}^{a'+b}(-\bar{y})^{a+b'}+x^{c'}(-y)^{d+d'}z^{c}\bar{x}^{a'+b}\bar{z}^{a+b'}\Big].
    \end{eqnarray}
     We have separated the above terms in 5 combinations. We show that each combination is an $R^{+}$. It is clear that these terms have the following structure.
    \begin{eqnarray}
    I(a,b,c,d,a',b',c',d')&=&x^{a}(-y)^{b}z^{c}\bar{x}^{a'}(-\bar{y})^{b'}\bar{z}^{c'}+x^{a}(-y)^{c}z^{b}\bar{x}^{a'}(-\bar{y})^{c'}\bar{z}^{b'}\nn\\\nonumber
    &+&x^{b}(-y)^{a}z^{c}\bar{x}^{b'}(-\bar{y})^{a'}\bar{z}^{c'}+x^{c}(-y)^{a}z^{b}\bar{x}^{c'}(-\bar{y})^{b'}\bar{z}^{a'}\\\nonumber
    &+&x^{b}(-y)^{c}z^{a}\bar{x}^{b'}(-\bar{y})^{c'}\bar{z}^{a'}+x^{c}(-y)^{b}z^{a}\bar{x}^{c'}(-\bar{y})^{b'}\bar{z}^{a'}.
    \end{eqnarray}
    where
    \begin{equation}\label{3m}
    I(a,b,c,a',b',c')=I(b,a,c,b',a',c')=I(c,b,a,c',b',a').
    \end{equation}
    In order to proceed we need to classify different orderings of $a,b,c$ and  $(a',b',c')$. In general, there are nine of them as follows
    \begin{eqnarray}
    1.\quad\quad a&\leq & b,c \quad\quad\quad a'\leq b' ,c',   \\\nonumber
    2.\quad\quad b&\leq &a , c\quad\quad\quad b' \leq  a', c', \\\nonumber
     3.\quad \quad c&\leq &a,b \quad\quad\quad c' \leq a', b',  \\\nonumber
    4.\quad \quad a&\leq &b , c\quad\quad\quad b' \leq  a', c', \\\nonumber
    5.\quad\quad  b&\leq & a,c \quad\quad\quad a'\leq b' ,c',   \\\nonumber
    6.\quad\quad  c&\leq & a,b \quad\quad\quad a'\leq b' ,c',   \\\nonumber
    7.\quad\quad  a&\leq &b ,c \quad\quad\quad c' \leq a', b',  \\\nonumber
    8.\quad\quad b&\leq &a ,c \quad\quad\quad c' \leq a', b',  \\\nonumber
    9.\quad\quad c&\leq &a,b\quad\quad\quad b' \leq  a', c', \\\nonumber
    \end{eqnarray}
    We first consider cases 1, 4 and 7.
    \begin{equation}\label{41m}
    1.\, I(a,b,c,a',b',c')=16^{a+a'}(-1)^{b+c}\sum_{k=0}^{c-a}\frac{(c-a)!}{k!(c-a-k)!}R^{+}_{b'-a',c'-a',b-a+k,c-a-k},
    \end{equation}
    \begin{equation}\label{51m}
    4.\  I(a,b,c,a',b',c')=(-16)^{a+b'}R^{+}_{c'-b',a'-b',c-a,b-a},
    \end{equation}
    \begin{equation}\label{61m}
    7.\  I(a,b,c,a',b',c')=(-16)^{a+c'}R^{+}_{b'-c',a'-c',b-a,c-a}.
    \end{equation}
    By using Eq.(\ref{3m}) and  Eq.(\ref{41m}) and switching  $\left(a\leftrightarrow b,a'\leftrightarrow b'\right)$, and $\left(a\leftrightarrow c,a'\leftrightarrow c'\right)$, one can resolve the  cases $2$ and $3$. Similarly, the cases 5 and 6  and the cases 8 and 9 can be obtained from the cases 4 and 7 respectively.
    \section{Basis for $R_{a,b,c,d}$}\label{apa}
      In this appendix we prove Eqs.\eqref{b19}-\eqref{b24}.

     By definition,
     \bea\label{b190}
    R_{a,b,0,1}&=&\bar{x}^a\left((-y)\bar{z}^b+z(-\bar{y})^b\right)+(-\bar{y})^a(x\bar{z}^b+z\bar{x}^b)\nn\\
    &+&\bar{z}^a\left((-y)\bar{x}^b+x(-\bar{y})^b\right).
    \eea
    Thus, the identity \eqref{g17} gives  Eq.\eqref{b19}.  Similarly,
         \bea\label{b200}
    R_{a,b,1,1}&=&x\bar{x}^a\left((-y)\bar{z}^b+z(-\bar{y})^b\right)+(-y)(-\bar{y})^a(x\bar{z}^b+z\bar{x}^b)\nn\\&+&z\bar{z}^a\left((-y)\bar{x}^b+x(-\bar{y})^b\right).
    \eea
    Therefore  Eq.\eqref{b20} is a result of the identity \eqref{g210}. Eq.\eqref{b21} can be verified by using Eq.\eqref{g21} in
    \bea\label{b210}
    R_{a,b,0,2}&=&\bar{x}^a\left(y^2\bar{z}^b+z^2(-\bar{y})^b\right)+(-\bar{y})^a(x^2\bar{z}^b+z^2\bar{x}^b)\nn\\&+&\bar{z}^a\left(y^2\bar{x}^b+x^2(-\bar{y})^b\right),
    \eea
    Eq.\eqref{g210} and Eq.\eqref{g170} give
       \be\label{z}
    z^4=z^2(\mathfrak{j}-xy)=\mathfrak{j} z^2-16z.
    \ee
    Therefore,
    \be\label{xy}
    x^2y^2=\mathfrak{j}^2-\mathfrak{j}z^2-16z.
    \ee
    Similarly,
    \be\label{xz}
    x^2z^2=\mathfrak{j}^2-\mathfrak{j}y^2+16y,
    \ee
    \be\label{yz}
    y^2z^2=\mathfrak{j}^2-\mathfrak{j}x^2-16x.
    \ee
    Using Eqs.\eqref{xy}-\eqref{yz} in
    \bea\label{b220}
    R_{a,b,2,2}&=& x^2\bar{x}^a\left(y^2\bar{z}^b+z^2(-\bar{y})^b\right)+y^2(-\bar{y})^a(x^2\bar{z}^b+z^2\bar{x}^b)\nn\\&+&z^2\bar{z}^a\left(y^2\bar{x}^b+x^2(-\bar{y})^b\right),
    \eea
    one obtains Eq.\eqref{b22}.    Finally, using the identity
    \be
    x\bar{x}^a\left[y^2\bar{z}^b+ z^2(-\bar{y})^b\right]=\bar{x}^a\left[(-y)\left(-\mathfrak{j}+z^2\right)\bar{z}^b+ \left(-\mathfrak{j}+y^2\right) z(-\bar{y})^b\right],
    \ee
    in
    \bea\label{b230}
    R_{a,b,1,2}&=& x\bar{x}^a\left(y^2\bar{z}^b+z^2(-\bar{y})^b\right)+(-y)(-\bar{y})^a(x^2\bar{z}^b+z^2\bar{x}^b)\nn\\&+&z\bar{z}^a\left(y^2\bar{x}^b+x^2(-\bar{y})^b\right),
    \eea
    one obtains Eq.\eqref{b24}.




 \end{document}